\UseRawInputEncoding
 
\documentclass[12pt]{article}

\usepackage{epsfig}

\usepackage{amssymb}
\usepackage{amsfonts}

\usepackage{color}
 
\def\be{\begin{equation}}
\def\ee{\end{equation}}
\def\ba{\begin{array}{c}}
\def\ea{\end{array}}

\newcommand{\bea}{\begin{eqnarray}}
\newcommand{\eea}{\end{eqnarray}}

\newtheorem{thm}{Theorem}

\newtheorem{lemma}[thm]{Lemma}

\newenvironment{proof}{\noindent
 {\bf Proof.}}{\hfill$\square$\vspace{3mm}\endtrivlist}

\begin{document}


\vspace{.35cm}

\begin{center}

{\Large

Exceptional points and domains of unitarity
for a class of strongly
non-Hermitian real-matrix
Hamiltonians

 }

\vspace{10mm}

\textbf{Miloslav Znojil}

\vspace{0.2cm}

The Czech Academy of Sciences, Nuclear Physics Institute,

 Hlavn\'{\i} 130,
250 68 \v{R}e\v{z}, Czech Republic

\vspace{0.2cm}

 and

\vspace{0.2cm}

Department of Physics, Faculty of Science, University of Hradec
Kr\'{a}lov\'{e},

Rokitansk\'{e}ho 62, 50003 Hradec Kr\'{a}lov\'{e},
 Czech Republic

\vspace{0.2cm}

{e-mail: znojil@ujf.cas.cz}

%
%
%
%
%
%
%
%
%
%
%
%
%
%
%
%
%
%
%
%
%
%
%

\end{center}



\section*{Abstract}

A family of non-Hermitian
real and tridiagonal-matrix
candidates
$H^{(N)}(\lambda)=H^{(N)}_0+\lambda\,W^{(N)}(\lambda)$
for a hiddenly Hermitian
(a.k.a. quasi-Hermitian)
quantum
Hamiltonian
is proposed and studied.
Fairly weak assumptions are imposed
upon
the unperturbed matrix
(the
square-well-simulating spectrum of
$H^{(N)}_0$ is
not assumed equidistant)
as well as upon its
maximally non-Hermitian $N-$parametric
antisymmetric-matrix perturbations
(matrix $W^{(N)}(\lambda)$ is
not even required to be ${\cal PT}-$symmetric).
In spite of that,
the ``physical'' parametric domain ${\cal D}^{[N]}$
is (constructively!) shown to exist, guaranteeing that
in its interior
the spectrum remains real and non-degenerate,
rendering the quantum evolution unitary.
Among the
non-Hermitian degeneracies
occurring at the boundary $\partial {\cal D}^{[N]}$
of the domain of stability
our main attention is paid to
their extreme version corresponding to the
Kato's
exceptional point of order $N$ (EPN).
The localization of the EPNs and, in their vicinity,
of the quantum-phase-transition
boundaries $\partial {\cal D}^{[N]}$
is found feasible, at the not too large $N$,
using computer-assisted symbolic manipulations
including, in particular, the Gr\"{o}bner basis elimination
and the high-precision arithmetics.

\newpage

\section{Introduction}


The currently observed enhancement of visibility
of non-Hermitian Hamiltonians $H\neq H^\dagger$
in quantum physics \cite{book}
was certainly motivated
by the Bender's and Boettcher's
conjecture \cite{BB}
that
the
bound states generated by
some ordinary differential non-Hermitian Hamiltonians
sampled by the imaginary cubic oscillator
 \be
 H^{(IC)}=p^2+ {\rm i}q^3
 \label{imcu}
 \ee
may be assigned,
in spite of their non-Hermiticity,
an entirely standard
probabilistic and unitary-evolution interpretation
-- see, e.g., the detailed explanation of such
an innovative possibility in
reviews \cite{Carl,ali}.

In 2012
Siegl with Krej\v{c}i\v{r}\'{\i}k \cite{Siegl}
revealed
(and, in 2019,
G\"{u}nther with Stefani \cite{Uwe} reconfirmed)
that, unfortunately, the ordinary differential
operator $H^{(IC)}$
which is non-Hermitian in $L^2(\mathbb{R})$ but
which was conjectured
to be Hermitizable (via an {\it ad hoc\,} amendment
of the inner product \cite{Carl,ali})
cannot in fact be Hermitized (i.e., consistently
interpreted) at all.
This means that the unbounded operator (\ref{imcu}) cannot
serve as an illustrative hiddenly Hermitian
example anymore.
Indeed, the boundedness of the non-Hermitian
candidates for Hamilotnians
seems to be, in the
mathematically consistent quantum theories,
technically useful if not even
necessary \cite{Dieudonne,Geyer}.

New benchmark models have to be sought, therefore.
In a way inspired by papers~\cite{maximal}
we are going to propose and study here one of the
eligible families of the candidates for such
Hamiltonians prescribed
in an $N$ by $N$ tridiagonal real matrix form
$H^{(N)}(\lambda)$.
They will be all assumed
decomposed into an exactly solvable
(i.e., in practice, diagonal and $\lambda-$independent)
``unperturbed''
component $H^{(N)}_0$ and its $\lambda-$dependent
``perturbation'',
 \be
 H^{(N)}(\lambda)=
 H_0^{(N)}+\lambda\,W^{(N)}(\lambda)\,.
 \label{myha}
 \ee
For methodical reasons the
perturbation will be assumed antisymmetric,
$\left [W^{(N)}\right ]^T(\lambda)=-W^{(N)}(\lambda)$,
rendering its
non-Hermiticity, at any $\lambda$,
maximal.

In contrast to papers~\cite{maximal}
or studies \cite{Muga}
we will {\em not\,} impose any other,
auxiliary form of
symmetry
lowering the number of
independently variable matrix elements.
In particular,
no form of
${\cal PT}$ symmetry
will be postulated
(meaning, in its real-matrix implementation
in \cite{maximal} for example, an additional
symmetry of the matrix
with respect to its second diagonal).
In this sense the flexibility of
our perturbation matrices
 \be
 \lambda\,W^{(N)}(\lambda)
 =\left[ \begin {array}{ccccc}
 0&W_1(\lambda)
  &0&\ldots&0
  \\{}-W_1(\lambda)&0
  &\ddots
 &\ddots&\vdots
 \\{}0&-W_2(\lambda)
 &\ddots&W_{N\!-\!2}(\lambda)&0
 \\{}\vdots&\ddots&\ddots&0&W_{N\!-\!1}(\lambda)
 \\{}0&\ldots&0&-W_{N\!-\!1}(\lambda)&0
 \end {array} \right]\,
 \label{entoy}
 \ee
containing an
$(N-1)-$plet of real and independent
functions of $\lambda$
will remain unrestricted.

One of the most interesting
innovations provided by the
quasi-Hermitian (QH, \cite{Geyer})
upgrade of the
theory lies in
a smooth parameter-mediated
access
to exceptional points
(EPs, \cite{Kato}).
In these limits
operators $H^{(N)}(\lambda^{(EP)})$
cease to be
diagonalizable so that
the quantum system itself ceases to be observable.
For this reason
one always has to specify
the
``unitarity-supporting''
domain ${\cal D}$ of admissible parameters
at which the spectrum remains real and non-degenerate.
In the language of physics this means that
at the EP boundary $\partial {\cal D}$
one encounters
the phenomenon of
quantum phase transition \cite{Geyer,denis}.

One of the
first illustrations of such a phenomenon
was provided by the
Bender's and Boettcher's
Hamiltonians
$H(\lambda)=p^2+ ({\rm i}q)^\lambda q^2$.
The
often-cited picture of the real and discrete spectrum
of these
contradictory models
exhibits, indeed, a clear
indication of the EP-related quantum-phase-transition
behavior
at $\lambda=0$ \cite{BB}.
Similar fructifications of the innovative
mathematical ideas in applications is also currently finding
new extensions reaching beyond the unitary quantum
systems \cite{Nimrod} and including even some
non-quantum, classical parts of physics \cite{Carlbook}
as well as their truly sophisticated nonlinear
versions and experimental
consequences \cite{Christodoulides}.

In all of these highly promising
directions of research
the above-mentioned disproofs of the
QH nature of the Bender's and Boettcher's popular
unbounded-operator models
(complemented also by
their hypersensitivity
to perturbations \cite{Trefethen,Viola})
reopened the problem of a
mathematically consistent QH-based
description of
the quantum phase transitions.
The search for an alternative EP-supporting
benchmark model had to be reopened.
Various alternative
matrix models $H^{(N)}(\lambda)$ with finite $N$
were introduced to serve the purpose
(cf., e.g., \cite{FR,BH}).
In our present paper a new family
of such models will be proposed and analyzed.


\section{The square-well-like choice of unperturbed spectrum}

During the search for less contradictory applications
of the QH formalism,
an efficient defence against the mathematical
criticism has been found in
a return to the older, more restrictive
QH formulation of
quantum mechanics \cite{Geyer}.
In this version of the theory
all of the non-Hermitian
but potentially Hermitizable
candidates for quantum Hamiltonians
had to be bounded operators.
Naturally, such a constraint
seems over-restrictive because it
disqualifies
a number of
non-Hermitian operators with real spectra.
At the same time,
the bounded-operator QH approach strongly
supports the study of various
parameter-dependent
$N$ by $N$ matrices  $H^{(N)}(\lambda)$
with a clear phenomenological appeal
-- see, e.g.,
the reasonably realistic toy-model in \cite{Geyer},
or the older non-Hermitian
manybody
Hamiltonians used by Dyson \cite{Dyson}.

In this framework
the motivation of our present study dates back to
our older
papers~\cite{maximal}
where the unperturbed spectrum
has been chosen equidistant.
For the purely
technical reasons, we also postulated there a few
other {\it ad hoc\,} constraints which
simplified the mathematics but which, at the same time,
restricted, severely,
the phenomenological appeal of the
resulting models.
In our present continuation of these efforts
aimed at a better understanding
of the strongly non-Hermitian models
with real spectra
we will remove some of these formal constraints.

In our $N-$numbered sequence of non-Hermitian toy models (\ref{myha})
the first item is
represented by the most elementary one-parametric
two-by-two Hamiltonian matrix. In an amended notation we redefine
$\lambda\,W_1(\lambda) \to \sqrt {A(\lambda)}$
and by
shifting the
arbitrarily scaled
unperturbed spectrum we make the matrix traceless,
 $$
 \left[ \begin {array}{cc} 1&\lambda\,W_1
 \\\noalign{\medskip}-\lambda\,W_1 &
3\end {array} \right]\ \to  \
 \widetilde{H^{(2)}}=
 \left[ \begin {array}{cc} 1&\sqrt {A}
 \\\noalign{\medskip}-\sqrt {A}&
3\end {array} \right]
\ \to \
 H^{(2)}=
 \left[ \begin {array}{cc} -1&\sqrt {A}
 \\\noalign{\medskip}-\sqrt {A}&
1\end {array} \right]\,.
 $$
The unperturbed spectrum
remains formally equidistant so that
the discussion of this case is still covered by
Ref.~\cite{maximal} and may be skipped. We only note that
for the real parameter $A$
the physical admissibility domain ${\cal D}^{[2]}$ is a
semi-infinite
interval $(-\infty,1)$.
Inside this domain
the pair of the bound state energies
$E_\pm(A) =\pm \sqrt{1-A}$ is real and non-degenerate,
forming the pair of branches of a parabola.
The single finite
boundary point is the Kato's
exceptional point of order two,
$A=A^{(EP2)} \in \partial {\cal D}^{[2]}=\{1\}$.
For our present methodical purposes we will mostly consider
the values of $A>0$ (i.e., the real and manifestly non-Hermitian
Hamiltonians), keeping in mind that at $A\leq 0$
the model would become conventional, Hermitian
and uninteresting.

Analogous modifications will also apply to
the general Hamiltonians of our present interest,
 $$
 \widetilde{H^{(N)}}=
 \left[ \begin {array}{ccccc}
 \widetilde{E}_{1}&\sqrt {A}&0&\ldots&0
 \\\noalign{\medskip}-
 \sqrt {A}&\widetilde{E}_{2}&\sqrt {B}&\ddots&\vdots
 \\\noalign{\medskip}0&-\sqrt {B}&\widetilde{E}_{3}^{(N)}&\ddots
 &0\\\noalign{\medskip}\vdots&\ddots&\ddots&\ddots&\sqrt {Z}
 \\\noalign{\medskip}0
 &\ldots&0&-\sqrt {Z}&\widetilde{E}_{N}\end {array} \right]
 \ \to \
 H^{(N)}=
 \left[ \begin {array}{ccccc}
 {E}_{1}^{(N)}&\sqrt {A}&0&\ldots&0
 \\\noalign{\medskip}-
 \sqrt {A}&{E}_{2}^{(N)}&\sqrt {B}&\ddots&\vdots
 \\\noalign{\medskip}0&-\sqrt {B}&{E}_{3}^{(N)}&\ddots
 &0\\\noalign{\medskip}\vdots&\ddots&\ddots&\ddots&\sqrt {Z}
 \\\noalign{\medskip}0
 &\ldots&0&-\sqrt {Z}&{E}_{N}^{(N)}\end {array} \right]\,.
 $$
The ``untilding''
(mediated by a suitable
shift of the energy scale)
is again assumed to make the Hamiltonian
traceless.

We are now encountering the problem of choice of
a suitable
non-equidistant unperturbed spectrum.
In this respect our main inspiration is given by the
ordinary differential square-well models
in which such a spectrum is non-equidistant and given,
typically, by a
quadratic polynomial
 \be
\widetilde{E}_n^{} = \widetilde{E}_1^{}+(n-1)\,c^{}_1
 +(n-1)^2\,c^{}_2 \,,
 \ \ \ \ n = 1, 2, \ldots\,
 \label{SQW}
 \ee
where, say, $\widetilde{E}_1^{}=1$.
For the sake of definiteness let us, therefore,
postulate that also here, the
unperturbed-energy elements will be square-well-like.
In this setting we decided to
define them
by the same polynomial~(\ref{SQW}), i.e., by the recurrences $\widetilde{E}_{n}=\widetilde{E}_{n-1}+\delta_n$
with a suitable linear function $\delta_n$ of the subscript $n$.
For the sake of definiteness
we picked up
$\delta_n =3n-4$.
This choice defines
the coefficients in (\ref{SQW}),
$$
\widetilde{E}_{n}=1+\frac{(n-1)(3n-2)}{2}\,,\ \ \ \ n=1, 2, \ldots\,.
$$
These energies (sampled in
Table \ref{xp4}) will specify
the ``unperturbed'', diagonal matrix elements in
Hamiltonians~(\ref{myha}).

\begin{table}[h]
\caption{Recurrently defined sequences of the unshifted (tilded)
 square-well-like \newline
 \mbox{\ \ \ \ \ \ \ \ \ \ \ } energies
 $\widetilde{E}_{n}=\widetilde{E}_{n-1}+\delta_n$
and of their $n-$plet averages
    of Eq.~(\ref{dj}).}
\label{xp4}

\vspace{2mm}

\centering
\begin{tabular}{||c||rrrrrrrrrr||}
\hline \hline
$n$&1&2&3&4&5&6&7&8&9&10\\
$\widetilde{E}_{n}$&1&3&8&16&27&41&58&78&101&127\\
\hline
$\delta_n $&-1&2&5&8&11&14&17&20&23&26\\
   $\Delta_n$    &1&2&4&7&11&16&22&29&37&46\\
 \hline
 \hline
\end{tabular}
\end{table}

Naturally, the resulting tilded non-Hermitian Hamilotians are not traceless.
Still, our choice of the unperturbed spectrum has an advantage of
making their trace equal to an integral multiple of their dimension.
The untilding shift of the spectrum is given by an integer, therefore,
 \be
  E_n^{(N)}=\widetilde{E}_n-\Delta_N\,,
 \ \ \ \
 \Delta_N=\frac{1}{N}\,\sum_{n=1}^N\,\widetilde{E}_n=\Bigl( \!\!\!\ba
   \vspace{-.2cm}
   N\\
   2
   \ea\!\!\!
   \Bigr) +1\,,
   \ \ \ \
 \sum_{n=1}^N\,E_n^{(N)}=0\,.
 \label{dj}
 \ee
The ultimate traceless Hamiltonians are
characterized by the
compact formula for the unperturbed energies
 \be
 E_n^{(N)}=\frac{1}{2}\,(-N^2+N+2-5\,n+3\,n^2)\,
 \ \ \ \ n=1,2,\ldots,N\,.
 \label{neporene}
 \ee
These values are integers
ranging from the negative ground state energy
${E}_{1}^{(N)}=-\Bigl( \!\!\!\ba
   \vspace{-.2cm}
   N\\
   2
   \ea\!\!\!
   \Bigr) $
up to the positive highest excitation
${E}_{N}^{(N)}=(N-1)^2$ (cf. Table \ref{yp4}).

Having completed the specification of our present class
of Hamiltonians, we are now prepared to pick up
the separate matrix dimensions $N=3,4,\ldots$
and
to study the respective finite-dimensional
perturbed square-well-like
models.

\begin{table}[h]
\caption{Unperturbed square-well-like
spectra (\ref{neporene})}\label{yp4}
\vspace{2mm}
\centering
\begin{tabular}{||c||rrrrrrr||}
\hline \hline
$n$&1&2&3&4&5&6&7\\
$N$&&&&&&&\\
\hline
$2$&-1&1&&&&&\\
$3$&-3&-1&4&&&&\\
$4$&-6&-4&1&9&&&\\
$5$&-10&-8&-3&5&16&&\\
$6$&-15&-13&-8&0&11&25&\\
$7$&-21&-19&-14&-6&5&19&36\\
 \hline
 \hline
\end{tabular}
\end{table}

\section{The search for maximal admissible perturbations\label{setri}}

The
maximal admissible strength $\lambda_{\max}$
of every
unitarity-non-violating perturbation
$\lambda\,W^{(N)}(\lambda)$ will
depend on the choice of the unperturbed Hamiltonian
$H_0^{(N)}$ (which we made above).
This being determined and fixed,
the value $\lambda_{\max}$ will also vary
with the
broad menu of choices of the $(N-1)-$plet of
the matrix-element functions
$\lambda\,W_1(\lambda)=\sqrt{A(\lambda)},
\lambda\,W_2(\lambda)=\sqrt{B(\lambda)}, \ldots$.
In a way guided by the results of \cite{maximal}
we may expect that
the corresponding values of $\lambda_{\max}$
will mark the boundaries $\partial {\cal D}^{(N)}$
of the domain of unitarity
characterized by the Kato's
exceptional-point energy-level
degeneracies of various orders.
In what follows we will only be interested
in the  maximal admissible perturbations,
i.e., in the
localization of the Kato's EP
degeneracies of the maximal order equal to dimension $N$ (EPN).

\subsection{EPN parameters at $N=3$}

The
non-equidistance of the unperturbed spectrum
is most easily illustrated using
the two-parametric and tridiagonal real-matrix model
with $N=3$.
This model
exemplifies also the shift of the energy scale
and introduction of our simplest nontrivial
traceless matrix Hamiltonian,
 \be
 \widetilde{H^{(3)}}(A,B)=\left[ \begin {array}{ccc} 1&\sqrt {A}&0
 \\\noalign{\medskip}-\sqrt {A}&3&\sqrt {B}
 \\\noalign{\medskip}0&-\sqrt {B}&8\end {array} \right]
 \ \to \ H^{(3)}(A,B)=\left[ \begin {array}{ccc} -3&\sqrt {A}&0
 \\\noalign{\medskip}-\sqrt {A}&-1&\sqrt {B}
 \\\noalign{\medskip}0&-\sqrt {B}&4\end {array} \right]\,.
 \label{nehje3}
 \ee
The analysis of the related
secular equation
 \be
 {E}^{3}+ \left( -13+A+B \right)\,E+3\,B-4\,A-12=0\,
 \label{secue3}
 \ee
proves already sufficiently instructive.
First of all,
the search for a maximal admissible,
reality-of-spectrum preserving perturbation
(cf. \cite{maximal})
remains
elementary. Indeed, the
reduction of Eq.~(\ref{secue3}) to its
exceptional-point-determining form
$E^3=0$ may be described by the mere coupled
pair of the linear algebraic
equations $A+B=13$ and $3\,B-4\,A=12$.
They have
the unique and exact real solution
 $$
 A^{(EP3)}=27/7\,,\ \ \ \ \
  B^{(EP3)}=64/7\,
 $$
where the acronym EP3 stands for the Kato's
exceptional point of order three.


A backward insertion of the EP3-representing parameters
in $H^{(3)}(A^{(EP3)},B^{(EP3)})$
reveals that this matrix becomes non-diagonalizable.
The EP3 nature of the singularity (with the geometric
multiplicity one) becomes reconfirmed via the corresponding
generalized Schr\"{o}dinger equation
 \be
  H^{(3)}(A^{(EP3)},B^{(EP3)})\,Q^{(3)}=Q^{(3)}\,J^{(3)}\,,
  \ \ \ \
  J^{(3)}= \left[ \begin {array}{ccc}
  0&1&0
  \\\noalign{\medskip}0&0&1
 \\\noalign{\medskip}0&0&0\end {array} \right]
 \ee
where the solution
 \be
 Q^{(3)}=\frac {1}{7}\, \left[ \begin {array}{ccc} { {36}}&-21&7
 \\\noalign{\medskip}{
 {12}}\,\sqrt {21}&-3\,\sqrt {21}&0
\\\noalign{\medskip}{
 {24}}\,\sqrt {3}&0&0\end {array} \right]\,
 \ee
is called transition matrix (useful,
e.g., in the context of
constructive perturbation theory \cite{admissible}).


\subsection{EPN parameters at $N=4$}

For the three-parametric model
 $$
H^{(4)}=\left[ \begin {array}{cccc}
  -6&\sqrt {A}&0&0
 \\\noalign{\medskip}-\sqrt {A}&-4&\sqrt {B}&0
 \\\noalign{\medskip}0&-\sqrt {B}&1&\sqrt {C}
\\\noalign{\medskip}0&0&-\sqrt {C}&9\end {array} \right]
 $$
the reduction of
secular equation
 $$
{E}^{4}+ \left( -67+C+B+A \right) {E}^{2}+ \left( -3\,B-150-10\,A+10\,
C \right)\,E+AC-54\,B+216+24\,C+9\,A=0
 $$
to its EP4 version $E^4=0$
yields the triplet
of nonlinear algebraic
sufficient conditions
 \be
 A+B+C=67\,,
 \ \ \ \
  -3\,B-10\,A+10\,C =150\,,
  \ \ \ \
  AC-54\,B+24\,C+9\,A=-216\,.
 \ee
Fortunately, the Gr\"{o}bner-basis elimination technique
reduces the problem to the solution
of quadratic equation
 $$
 7\,{A}^{2}+1716\,A-16848=0\,
 $$
with the unique and exact positive
root
 $$
 A^{(EP4)}=-{\frac {858}{7}}+{\frac {30}{7}}\,\sqrt {949}
 \approx 9.4536\,.
 $$
The same elimination technique also produces
the
two exact and linear, explicit
definitions of the remaining EP4 coordinates,
 $$
 13\,B=520-20\,A\,,\ \ \ \ \ 13\,C=351+7\,A\,.
 $$
This implies that in the EP4 limit
our Hamiltonian is the real matrix with positive
values of the coupling constants
 $$
 B^{(EP4)}={\frac {1600}{7}}-{\frac {600}{91}}\,\sqrt {949}
 \approx 25.4560\,
 $$
and
 $$
 C^{(EP4)}=-39+{\frac {30}{13}}\,\sqrt {949}
 \approx 32.0904\,.
 $$
The consequences and, in particular, the maximality of the
perturbation compatible with the reality of the spectrum remain
the same as at $N=3$. At the same time, the next-step transition
to $N=5$ will already be a partially numerical task. For the
analysis and description of the spectrum a computer-assisted symbolic manipulation
becomes warmly welcome.

\section{Computer-assisted constructions\label{sepeti}}

\subsection{EPN parameters at $N=5$}




The choice of $N=5$ with four parameters in
 $$
H^{(5)}=
 \left[ \begin {array}{ccccc}
 -10&\sqrt {A}&0&0&0
 \\\noalign{\medskip}-
\sqrt {A}&-8&\sqrt {B}&0&0
\\\noalign{\medskip}0&-\sqrt {B}&-3&\sqrt {C
}&0\\\noalign{\medskip}0&0&-\sqrt {C}&5&\sqrt {D}
\\\noalign{\medskip}0
&0&0&-\sqrt {D}&16\end {array} \right]
 $$
leads to the secular equation
 $$
{E}^{5}+ \left( -227+C+B+A+D \right) {E}^{3}
+\left( 21\,D-18\,A-11\,B
-894+2\,C \right) {E}^{2}+
$$
$$
+ \left( 134\,D-208\,C-130\,B+AC+5680+BD+AD+
17\,A \right) E+
$$
$$+19200+240\,A+800\,B+10\,BD+3\,AD
+240\,D-16\,AC-1280\,C=0\,.
 $$
The localization of the EP5 coordinates
is now mediated by the quadruplet
of coupled algebraic equations
 $$
 C+D+A+B-227=0\,,\ \ \ \ \
21\,D-18\,A-11\,B-894+2\,C=0\,,\ \ \ \ \
$$
$$
134\,D-208\,C-130\,B+AC+5680+BD+AD+17\,A=0\,,\ \ \ \ \
$$
$$
19200+240\,A+800\,B+10\,BD
+3\,AD+240\,D-16\,AC-1280\,C=0\,. \
 $$
The Gr\"{o}bner-basis elimination leads to the following implicit
definition of $B$,
 \be
 41405\,{B}^{4}+42197064\,{B}^{3}
 -35083975824\,{B}^{2}-23755497730560\,B
 +1288938668811264=0\,.
 \label{determe}
 \ee
This equation (still exactly solvable, in principle at least)
is accompanied by the three elementary
definitions of the
remaining three EP5-determining coupling constants, viz.,
 $$
 -5345760274992\,B+10269063168\,{B}^{2}+32585735\,{B}^{3}+
38816530841664\,A-484377515547264=0\,,
 $$
 $$
 245838028663872\,C-423614555\,{B}^{3}-133497821184\,{B}^{2}+
 $$
 $$
 +
483537879219312\,B-43815233614473792=0\,,
  $$
$$
184378521497904\,D+162928675\,{B}^{3}+51345315840\,{B}^{2}-
152882526610368\,B-6691705970319360=0\,.
 $$
Among the quadruplet
of the approximate numerical
roots of Eq.~(\ref{determe}), viz.,
$$
  -1318.1571\, \ \ \ \ -569.5091\,,
  \ \ \ \  50.7046\,, \ \ \ \ 817.8319
 $$
only the latter two positive values remain compatible with the
reality of the Hamiltonian matrix. As long as the fourth
value of $B$ would lead to the unacceptable (i.e.,
negative) $A= -511.0383$, the
localization of
the eligible EP5 values is unique, with
 $$
 A^{(EP5)}=18.6720\,,\ \ \
 B^{(EP5)}=50.7046\,,\ \ \
 C^{(EP5)}=80.1181\,,\ \ \
 D^{(EP5)}=77.5053\,.
 $$
As long as such a result
is based on the knowledge of the roots of Eq.~(\ref{determe}),
the localization of the EP5 singularity
may be treated as
exact.
Nevertheless,
an elementary computer-assisted experiment may be performed to
show that
the exact representation of the roots of the above-mentioned
quartic polynomial (\ref{determe}) would occupy many printed pages
and is, therefore, of
a literally zero practical usefulness.

\subsection{EPN parameters at $N=6$}

Hamiltonian
 \be
H^{(6)}=
 \left[ \begin {array}{cccccc} -15&\sqrt {A}&0&0&0&0
 \\\noalign{\medskip}-\sqrt {A}&-13&\sqrt {B}&0&0&0
 \\\noalign{\medskip}0
 &-\sqrt {B}&-8&\sqrt {C}&0&0
 \\\noalign{\medskip}0&0&-\sqrt {C}&0&
 \sqrt {D}&0
 t\\\noalign{\medskip}0&0&0&-\sqrt {D}&11&\sqrt {F}
 \\\noalign{\medskip}0&0&0&0&-\sqrt {F}&25
 \end {array} \right]
 \label{mo6}
 \ee
and its
secular equation
 $$
 {s}^{6}+ \left( -602+F+D+C+B+A \right) {s}^{4}+ \left( 36\,F-3624+11\,
D-8\,C-21\,B-28\,A \right) {s}^{3}+
 $$
 $$
 +\left( 419\,F+FC+FB+FA-481\,D-538
\,C+59065-265\,B-13\,A+DB+DA+CA \right) {s}^{2}+
 $$
 $$
 +\left( 1560\,F+28\,FC
+15\,FB+8\,FA-8915\,D-10\,DB-17\,DA+680\,C+4125\,B+429000+\right .
 $$
 $$
 \left .
 +2200\,A-36\,
CA \right) s+53625\,C-39000\,D+275\,CA-375\,DB-200\,DA+FCA+195\,FC=0
 $$
may be treated along the same lines as their $N=5$ predecessors.
Using an analogous symbolic manipulation procedure one finds that
the $N=6$ analogue of the Gr\"{o}bner-reduced $N=5$
polynomial of Eq.~(\ref{determe})
is of the 11th degree, characterized also
by a truly excessive
length of its integer coefficients. The use of
the high precision arithmetics is recommended to
provide the reliable
numerical values
 $$
A^{(EP6)}= 32.3950\,,\ \ \ \
 B^{(EP6)}=86.6542\,, \ \ \ \
C^{(EP6)}= 146.7324\,
$$
 $$
\ \ \ \
D^{(EP6)}= 183.1682\,, \ \ \ \
 F^{(EP6)}= 153.0502\,
 $$
of the EP6 coordinates. In spite of being only
stored in the computer,
the available enhanced precision of these solutions will be
shown to lead to remarkable consequences
in what follows.

\section{Domains of the spectral reality \label{sectyri}}


\subsection{Domain ${\cal D}^{[N]}$ at $N=3$}

The change
$A=27/7-m$ and $B=64/7-n$ of the coordinates in
the two-dimensional physical-admissibility domain
${\cal D}^{[2]}$ yields the modified Hamiltonian
 \be
 H^{(3)}[m,n]=\left[ \begin {array}{ccc} -3&\sqrt {27/7-m}&0
 \\\noalign{\medskip}-\sqrt {27/7-m}&-1&\sqrt {64/7-n}
 \\\noalign{\medskip}0&-\sqrt {64/7-n}&4\end {array} \right]
 \label{naha3}
 \ee
as well as another transparent form
of the secular equation,
 $$
 {E}^{3}+ \left( -n-m \right)\,E+4\,m-3\,n=0\,.
 $$
Its graphical interpretation is elementary:
the reality of its three roots requires the
existence of the two real extremes
of the
secular-polynomial curve at $E_\pm=\pm \sqrt{(m+n)/3}$.
This means that the
quantity under the square root
must be positive,
$m+n=3\,p > 0$. One may eliminate  $m=3\,p-n$
and get the ultimate form of the secular equation
 $${E}^{3}-3\,p\, E+2\,q=0$$
where
 $2\,q=12\,p-7\,n$
just shifts the graph of the elementary curve ${E}^{3}-3\,p\, E$.
We have to guarantee that this shift is
such that
the curve has strictly three
(non-degenerate) intersections with the real line.
Thus, it is sufficient to evaluate the curve ${E}^{3}-3\,p\, E$
at to localize its minimum and maximum at $E_{\pm}= \pm \sqrt{p}$.
This forms an essence of the proof of the following result.

 \begin{lemma}
 \label{lemma1}
The spectrum of non-Hermitian Hamiltonian (\ref{nehje3})
is real and non-degenerate if and only if
 \be
 \frac{2p}{7}(6-\sqrt{p}) < n < \frac{2p}{7}(6+\sqrt{p})\,,\ \ \ \ \
 p>0\,.
 \label{leraus}
 \ee
The lower and upper bound would correspond to
the degeneracy
of the upper or lower doublet of levels, respectively.
\end{lemma}
\begin{proof}
The differentiated form of the
secular polynomial $\phi(x)={x}^{3}-3\,p\, x+2\,q$
may be factorized, $\phi'(x)=3\,{x}^{2}-3\,p=3\,(x-a)\,(x+a)$.
The new real parameter $a$ is arbitrary (or, better, non-negative)
and such that $3\,p=m+n=3\,a^2$
and $\max\, \phi(x)=\phi(-a)$ while  $\min\, \phi(x)=\phi(a)$.
Thus, polynomial $\phi(x)$ has three non-degenerate real roots
if and only if $q$ is such that $\phi(-a)>0$ and $\phi(a)<0$.
\end{proof}

One of the remarkable properties of formula (\ref{leraus})
is that its validity is not restricted to the vicinity of EP3
because the value of $p>0$ is arbitrary.  Far from EP3
but still inside domain
${\cal D}^{[2]}$
the shift $m$ of coupling $A$ may even be
negative.
This would only require that the value of $p$
becomes sufficiently large.
It is easy to see that one just would need to have
%
%
$n>3p$ and
$ p>81/4=20.25$. At an even larger
$ p>36$ the
values of the second shift $n$ of
coupling $B$ may be negative as well.

Both of these observations
concerning the systems with large $p$ are just a curiosity.
In the dynamically most interesting vicinity of the EP3
singularity our attention remains restricted to the
not too large parameters $p$. In fact,
what we require is that all of
our toy model matrices
(\ref{naha3}) remain real and non-Hermitian, i.e.,
 \be
 m< 27/7 = 4-1/7
 \ \ \ \ {\rm and } \ \ \ \ n<64/7 = 9+1/7\,.
 \label{boin}
 \ee
Under these assumptions
we have the constraint $p < 13/3 \approx 4.33 $.
At the larger values of $p$, at least
one of the inequalities of Eq.~(\ref{boin}) must
fail to hold rendering
the matrix complex,
i.e., incompatible with our assumptions.
Immediately  beyond
the critical $p_{\max} =13/3$ boundary,
one can even achieve a
{\em simultaneous\,} violation of {\em both\,} of the
inequalities. Paradoxically,
our $N=3$ model would then be strictly Hermitian.


\subsection{Domain  ${\cal D}^{[N]}$ at $N=4$}

Proceeding
via analogy with preceding subsection we introduce
the three real and positive shifts $m$, $n$ and $v$.
Their
subtraction from the EP4 constants yields
the new matrix
$H^{(4)}[m,n,v]$ with the reparametrized
off-diagonal elements
 $$
 H^{(4)}_{1,2}=-H^{(4)}_{2,1}=\frac{1}{7}\,\sqrt {-6006+210\,\sqrt {949}-49\,m}\,,
 $$
 $$
 H^{(4)}_{2,3}=-H^{(4)}_{3,2}={\frac {1}{91}}\,
 \sqrt {1892800-54600\,\sqrt {949}-8281\,n}\,
 $$
and
 $$
 H^{(4)}_{3,4}=-H^{(4)}_{4,}=\frac{1}{13}
 \,\sqrt {-6591+390\,\sqrt {949}-169\,v}\,.
 $$
In terms of the new parameters the exact form of
secular equation looks simpler than expected,
 \be
 {E}^{4}+ \left( -v-n-m \right) {E}^{2}+ \left( 3\,n+10\,m-10\,v
 \right)\,E-
 \label{ecue4}
 \ee
 $$-{\frac {30}{7}}\,\sqrt {949}v+54\,n+{\frac {690}{7}}\,v+30
\,m+mv-{\frac {30}{13}}\,m\sqrt {949}=0\,.
 $$
In
the secular polynomial
 $
 \phi^{(4)}(x)=
 {x}^{4}- 6 p {x}^{2}+ 8 q {x} + r\,
 $
of
Eq.~(\ref{ecue4})
the change of sign of $E$ would merely change the sign
of $q$.
Thus,
without any loss of generality
our study of the conditions
of the spectral reality
may be reduced to the subcase with non-negative $q\geq 0$.
The double-well shape of
the generic curve (\ref{ecue4})
(which would be left-right symmetric
at $q=0$) could only become, during
a move to a non-vanishing $q>0$, lifted up/down
at the respective positive/negative $x$.
The central local maximum $x_0$ moves
slightly to the right,  $x_0=a>0$, and an analogous shift
also applies to the rightmost minimum of $\phi^{(4)}(x)$
at $x_1=a+b$ with positive lower-case $b>0$.

In the parametric domain of unitarity ${\cal D}^{[3]}$
the energy spectrum must be kept real and non-degenerate.
This means that
the positive values of $a>0$ and of $b>0$
may be kept unconstrained, arbitrarily variable.
This is an important merit of the new lower-case
parameters. Their use
enables us to factorize the derivative of
the secular polynomial,
 \be
 [\phi^{(4)}]'(x)/4=
 {x}^{3}- 3 p {x}+ 2 q =(x-a)(x-a
 -b)(x+\gamma)\,,
 \ \ \ \ \ \gamma=2\,a+b > 0
 \,.
 \ee
We may evaluate
 \be
 3 p={b}^{2}+3\,a^{2}+3\,a\,b\,,\ \ \ \
 2q=2\,a^{3}+3\,a^{2}\,b+a\,{b}^{2}\,
 \label{threeolder}
 \ee
and
eliminate, via comparison with Eq.~(\ref{ecue4}), the two older shift parameters as redundant,
 \be
 v=v(n)=
 -\frac{1}{10}\,{a}^{3}-{\frac {3}{20}}\,{a}^{2}b-\frac{1}{20}\,a{b}^{2}+\frac{1}{2}\,{b}^{2}
 +\frac{3}{2}\,{a}^{2}+\frac{3}{2}\,ab-{\frac {7}{20}}\,n\,,
 \ee
 \be
 m=m(n)=
 \frac{1}{2}\,{b}^{2}+\frac{3}{2}\,{a}^{2}+\frac{3}{2}\,ab-{\frac {13}{20}}\,n+\frac{1}{10}\,{a}^{3}+{
 \frac {3}{20}}\,{a}^{2}b+\frac{1}{20}\,a{b}^{2}\,.
 \ee
These definitions convert
the constant term of Eq.~(\ref{ecue4})
in the closed-form $n-$dependent function
 $$
r=r(n)={\frac {91}{400}}\,{n}^{2}+ \left( 3\,\sqrt {949}-\frac{1}{2}\,{b}^{2}-\frac{3}{2}\,{a
}^{2}-\frac{3}{2}\,ab+{\frac {3}{100}}\,{a}^{3}+{\frac {9}{200}}\,{a}^{2}b+{
\frac {3}{200}}\,a{b}^{2} \right) n+
$$
$$
+{\frac {450}{7}}\,{b}^{2}+\frac{3}{2}\,{b}
^{3}a+{\frac {1350}{7}}\,ab+{\frac {1350}{7}}\,{a}^{2}+{\frac {27}{91}
}\,\sqrt {949}{a}^{2}b-{\frac {3}{100}}\,{a}^{5}b-{\frac {13}{400}}\,{
a}^{4}{b}^{2}+
$$
$$
+{\frac {9}{91}}\,\sqrt {949}a{b}^{2}-{\frac {900}{91}}\,
\sqrt {949}ab+{\frac {18}{91}}\,\sqrt {949}{a}^{3}-{\frac {300}{91}}\,
\sqrt {949}{b}^{2}+\frac{1}{4}\,{b}^{4}+\frac{9}{4}\,{a}^{4}+{\frac {15}{4}}\,{b}^{2}{
a}^{2}-{\frac {72}{7}}\,{a}^{2}b+
$$
$$
+\frac{9}{2}\,{a}^{3}b-{\frac {900}{91}}\,
\sqrt {949}{a}^{2}-{\frac {1}{400}}\,{a}^{2}{b}^{4}-{\frac {1}{100}}\,
{a}^{6}-{\frac {3}{200}}\,{a}^{3}{b}^{3}-{\frac {48}{7}}\,{a}^{3}-{
\frac {24}{7}}\,a{b}^{2}\,.
 $$
Our story comes to its climax.

\begin{lemma}
 \label{lemma2}
At  $N=4$ and at arbitrary preselected parameters $a>0$ and $b>0$
the reality and non-degeneracy of the $q\geq 0$ spectrum
is guaranteed by the two inequalities
 \be
 \phi^{(4)}(a+b)< 0 < \phi^{(4)}(a)
 \label{coni}
  \ee
which determine the admissible range of the shift parameter $n$.
\end{lemma}
\begin{proof}
In the representation
using the triplet of variable parameters $a$, $b$ and $n$,
the difference $\phi^{(4)}(x)-r(n)$ becomes $n-$independent.
As long as our secular polynomial $\phi^{(4)}(x)$
acquires its
local maximum at $x=a$ and its local minimum at $x=a+b$,
it is necessary and sufficient to require that
 \be
 r(n)-\phi^{(4)}(a+b) >r(n) >r(n)-\phi^{(4)}(a)\,.
 \label{necesa}
 \ee
 \end{proof}

 \noindent
The implicit exact definition  (\ref{coni})
of
the domain of admissible $n$s
may be most easily made explicit
in the regime of small parameters where
$r(n)$ is an approximately linear function of $n$.
A detailed inspection of our formulae then reveals that
the dominant  part of
inequalities (\ref{coni}) would
define the leading-order form of $n$ while their
subdominant part would specify, in
a partial parallel to our previous $N=3$ results,
the approximate
upper and lower bounds of
variability of the admissible,
unitarity-compatible
values of $n$. We will return to this idea later.

\subsection{Domain ${\cal D}^{[N]}$ at $N=5$}

Assuming that we have the secular polynomial
 \be
 \phi(E)=E^{\,5}+P\,E^{\,3}+Q\,E^{\,2}+R\,E+S
 \label{wcompb}
 \ee
let us follow the experience gained
in the preceding text and, in the first step,
let us identify the relevant
boundary-point elements of $\partial {\cal D}$
with the confluences of the zeros of the derivative
 \be
 \phi'(E)=5\,E^{\,4}+3\,P\,E^{\,2}+2\,Q\,E+R\,.
 \label{compb}
 \ee
In the most interesting
small vicinity of the EP5 limit
we will
again consider,
without any loss of generality, just the specific
$Q \geq 0$ scenario.
This will enable us to work
with the factorization of the derivative
 \be
  \phi'(x)/5=
 \left( x+c \right)  \left( x-a \right)  \left( x-a-b \right)  \left(x+2\,a+b-c \right)=
 \label{compa}
 \ee
  $$={x}^{4}+ \left( -3\,ab+2\,ca+cb-{b}^{2}-3\,{a}^{2}-{c}^{2} \right) {x}
^{2}+$$
 $$+ \left( 2\,{c}^{2}a-c{b}^{2}+3\,{a}^{2}b-4\,c{a}^{2}+a{b}^{2}+{c}
^{2}b+2\,{a}^{3}-4\,cab \right) x-$$
 $$-{c}^{2}{a}^{2}-{c}^{2}ab+3\,c{a}^{2}
b+2\,c{a}^{3}+ca{b}^{2}
 $$
where
all of the lower-case
parameters remain small, positive and
such that $c<a$.

In a way paralleling the $N=4$ analysis
the coefficients
$P=P(m,n,v,w)$ and $Q=Q(m,n,v,w)$
of Eq.~(\ref{compb})
have to be  reparametrized
via Eq.~(\ref{compa}), with
$P=P[m(abc),n(abc),v(abc),w(abc)]$ and with
$Q=Q[m(abc),n(abc),v(abc),w(abc)]$.
Due to the linearity of
$P=P(m,n,v,w)$ and $Q=Q(m,n,v,w)$
we may immediately
eliminate the two parameters,
 $$
 w=
 \left( -{\frac {2}{39}}\,a-\frac{1}{39}\,b+{\frac {6}{13}} \right) {c}^{2}+ \left( {\frac {4}{39}}\,{a}^{2}-{\frac {6}{13}}\,b+{\frac {4}{39}}\,a
b-{\frac {12}{13}}\,a+\frac{1}{39}\,{b}^{2} \right) c-
 $$
 $$
 -\frac{1}{39}\,a{b}^{2}+{\frac {
18}{13}}\,{a}^{2}-\frac{1}{13}\,{a}^{2}b+{\frac {18}{13}}\,ab-{\frac {20}{39}}
\,v-{\frac {7}{39}}\,n-{\frac {2}{39}}\,{a}^{3}+{\frac {6}{13}}\,{b}^{
2}
$$
 and
 $$
 m
 =\left( {\frac {7}{13}}+{\frac {2}{39}}\,a+\frac{1}{39}\,b \right) {c}^{2}+ \left( -\frac{1}{39}\,{b}^{2}-{\frac {7}{13}}\,b-{\frac {4}{39}}\,{a}^{2}-{
\frac {14}{13}}\,a-{\frac {4}{39}}\,ab \right) c+
$$
 $$
 +{\frac {21}{13}}\,ab+
{\frac {2}{39}}\,{a}^{3}-{\frac {32}{39}}\,n+{\frac {7}{13}}\,{b}^{2}+
{\frac {21}{13}}\,{a}^{2}+\frac{1}{13}\,{a}^{2}b-{\frac {19}{39}}\,v+\frac{1}{39}\,a{b
}^{2}\,.
 $$
The other coefficient
$R=R(m,n,v,w)$ is
quadratically non-linear
(i.e., still non-numerically tractable) function
of its arguments, so one must eliminate
$n$ or $v$.

Although the latter ambiguity
looks like a marginal technicality,
its weight would grow with
the increase of $N$
so let us now interrupt
the study of the
system without approximations.
Briefly, let us only add that
at
any $N \geq 6$ the study of the global shape
of ${\cal D}^{[N]}$
would remain feasible, leading just to
another, independent
series of applications of the Gr\"{o}bner-basis method
of solving the emerging coupled sets of the
nonlinear algebraic equations.
Incidentally, the display of results would again
exceed the capacity of
the current Journal's printed pages.
In other words,
our present recommendation would be to
restrict our forthcoming printable message just to
the description of
the
phenomenologically
most relevant part of ${\cal D}$
which lies in a small vicinity of
the dynamically most challenging
EP5 extreme.

For an introduction in such a setup let us
parallel the last two equations of preceding section
and let us omit
all of the higher order contributions
as inessential.
This linearizes
the remaining two
polynomial algebraic equations and leads, immediately,
to the leading-order and
numerical elementary
re-parametrization prescriptions
$$
 m=- 0.5042\,ac- 0.2521\,cb+
 0.7564\,{a}^{2}
 +0.7564\,ab+
 0.2521\,{b}^{2}+ 0.2521\,{c}^{2}\,,
 $$
 $$
 n=- 0.9295\,ac- 0.4647\,cb+
 1.3942\,{a}^{2}
 +1.3942\,ab+
 0.4647\,{b}^{2}+ 0.4647\,{c}^{2}
 \,,
 $$
 $$
 v=- 1.0838\,ac- 0.5419\,cb+
 1.6257\,{a}^{2}
 + 1.6257\,ab+
 0.5419\,{b}^{2}+ 0.5419\,{c}^{2}
 \,,
 $$
 $$
 w=- 0.8159\,ac- 0.4079\,cb+
 1.2238\,{a}^{2}
 + 1.2238\,ab+
 0.4079\,{b}^{2}+ 0.4079\,{c}^{2}
 \,.
 $$
One can easily check that as long as $c<a$,
all of these shifts are positive as expected.
At the same time, such an approach does not
provide any estimate of the influence of
corrections. The development of a different,
more systematic method quantifying
the role of corrections
is needed.

\section{Boundaries of corridors to exceptional points}

The construction and
description of the shape of boundary of
the $(N-1)-$dimensional
physical domain ${\cal D}^{[N]}$
may be expected to
become facilitated
in the phenomenologically
most relevant vicinity of the EPN
extreme.
For this purpose an amendment of the methods is needed.

\subsection{$N=3$: linear order-by-order equations}

Secular polynomial
$\phi^{(3)}(x)={x}^{3}+ \left( -n-m \right)\,x+4\,m-3\,n$ of
Hamiltonian (\ref{naha3})
and its factorized derivative
$\left[\phi^{(3)}\right ]'(x)=3\,{x}^{2}- \left( m+n \right)\,x=
3\,[(x-a)(x+a)]=3x^2-3a^2$ admit the reparametrization of
$m=m(a,n)=3\,a^2-n$. Under the
assumption of smallness of $x={\cal O}(\lambda)$ we may write
$x=\lambda\,\xi$, put $a=\lambda\,\alpha$, postulate
$n=\lambda^2\,\beta + \lambda^3\,\gamma+{\cal O}(\lambda^4)$
and reparametrize
the secular equation,
 \be
 {\lambda}^{}{\xi}^{3}
 - \left(3\,\alpha^2 -\beta - \lambda\,\gamma
 \right)\,x+
 12\,\alpha^2-7\,\beta - 7\,\lambda\,\gamma
 +{\cal O}(\lambda^2)
 =0\,.
 \ee
According to Lemma \ref{lemma1}  the minimum of
the polynomial is reached at $\xi=\alpha$ where we have
 \be
 - 2\,\lambda\,\alpha^3+
 12\,\alpha^2-7\,\beta- 7\,\lambda\,\gamma
 +{\cal O}(\lambda^2)
 =0\,
 \ee
while at its maximum with $\xi=-\alpha$  we have
 \be
 2\,\lambda\,\alpha^3+
 12\,\alpha^2-7\,\beta- 7\,\lambda\,\gamma
 +{\cal O}(\lambda^2)
 =0\,.
 \ee
In the two dominant orders
the construction degenerates to the pair of linear equations.
The leading-order
solution $\beta=12\,\alpha^2/7$ of these equations
exists and is unique.
In the next order these relations
specify the minimal and maximal
value of $\gamma$, restricting its variability
to the non-empty interval
$\, (\gamma_{-},\gamma_{+})\,$
where $\gamma_{\pm}=\pm 2\,\alpha^3/7+{\cal O}(\lambda)$.
Incidentally,
as long as we could recall Lemma~\ref{lemma1},
such an estimate
reproduces precisely the
exact nonperturbative result $\gamma_{\pm}=\pm 2\,\alpha^3/7$
of Eq.~(\ref{leraus}).

\subsection{$N=4$: linearizable order-by-order equations}

At $N=4$, secular  polynomial
 $$
 \phi^{(4)}(s)={s}^{4}+ \widetilde{P}\, {s}^{2}+ \widetilde{Q}\, s+\widetilde{R}=
 {s}^{4}+ \left( -n-m-v \right) {s}^{2}+ \left( 3\,n+10\,m-10\,v
 \right) s+\widetilde{R}\,
 $$
with the exact form of the constant term
 $$
 \widetilde{R}=54\,n+{\frac {690}{7}}\,v+30\,m-{\frac {30}{13}}\,m\sqrt {
 949}-{\frac {30}{7}}\,\sqrt {949}v+mv
 $$
enters the secular equation
 $
 \phi^{(4)}(s)=0$ of course.
We have to determine the three-dimensional ``physical'',
unitarity-compatibility domain ${\cal D}^{[4]}$
defined by the property
that
the triplet of the tilded parameters $\widetilde{P}$,
$\widetilde{Q}$ and $\widetilde{R}$ belongs
to  this domain if and only if
all of the four roots of $\phi^{(4)}(E)$ remain real
and non-degenerate.

The construction remains elementary
at $\widetilde{Q}=0$ since
the curve $\phi^{(4)}(x)$
is then symmetric so that
the
reality of the roots becomes equivalent to the
positivity of its local maximum $\phi^{(4)}(0)$
and to the negativity of its two equal values
$\phi^{(4)}(\pm b)$ at a suitable $x=b \neq 0$.
After we move to $\widetilde{Q} \neq 0$,
the local maximum leaves the origin
and moves to a non-zero
coordinate $x=a$.
Without any loss of generality we may assume that
$a \geq 0$ (which is correct for  $\widetilde{Q} \geq 0$)
because the same picture also covers
the other case after the mere
change of the sign of $x$.

Once we have $\widetilde{Q} \geq 0$,
the
second parameter $b$
determining the necessarily negative
value of  $\phi^{(4)}(x)$ at $x=a+b$ still
has to be
non-negative, $b \geq 0$. Its optimal choice should be made as
specifying the
relevant local
minimum of $\phi^{(4)}(x)$.
In this arrangement we
intend to convert
the exact implicit definition (\ref{necesa})
of the boundaries
of ${\cal D}^{[4]}$
as defined in Lemma \ref{lemma2}
into its approximate form
valid in a small vicinity of EP4.

The procedure is as follows.
First, setting
$x=\lambda\,\xi$, $a=\lambda\,\alpha$, $b=\lambda\,\beta$
and also
$ \widetilde{P}=\lambda^2\,\mu$ and $ \widetilde{Q}=\lambda^3\,\nu$
we rewrite our secular equation in
terms of the rescaled energy
$\xi$ and of
the two rescaled
parameters $\mu=\mu(\alpha,\beta)$ and $\nu=\nu(\alpha,\beta)$ of order
${\cal O}(\lambda^0)$,
 \be
 \widetilde{R}=
 -\lambda^{4}\,
 \left [{\xi}^{4}+ \mu\, {\xi}^{2}+
 \nu\, \xi\right ]\,
 \label{manine}
 \ee
As long as a naive estimate of
the left-hand-side
expression $ \widetilde{R}={\cal O}(\lambda^2)$
appears incompatible with the right-hand-side
order of smallness,
we
have to impose the constraint
that it vanishes in the two leading orders,
 \be
 \widetilde{R}=
 - 41.0904\,m+54\,n- 33.4536\,v
 +{\cal O}(\lambda^4)=0\,.
 \label{defr}
 \ee
This requirement may be realized via
postulates
 $$
 m=\lambda^2\,m_2+\lambda^4\,m_4
 \,,\ \ \
 n=\lambda^2\,n_2+\lambda^4\,n_4
 \,,\ \ \
 v=\lambda^2\,v_2+\lambda^4\,v_4
 $$
and by the split of the exact redefinitions
of the shifts,
i.e., of the two
exact linear relations between functions
 $$
 m+n+v=2\,
 \lambda^2\,\mu\,,\ \ \ \ \ \
 10\,m+3\,n-10\,v=\lambda^3\,\nu\,
 $$
into the two decoupled sets for the
leading-order coefficients,
 \be
 m_2+n_2+v_2=2\,
 \mu\,,\ \ \ \ \ \
 10\,m_2+3\,n_2-10\,v_2=\lambda\,\nu\,
 \label{ofeq}
 \ee
and for the subdominant-order coefficients,
 $$
 m_4+n_4+v_4=0\,,\ \ \ \ \ \
 10\,m_4+3\,n_4-10\,v_4=0\,.
 $$
This yields the one-parametric solution of the former set,
$$m_2=\mu+\lambda\,\nu/20 -13\,n_2/20=m_2(n_2)\,,
\ \ \ \ v_2=\mu-\lambda\,\nu/20 -7\,n_2/20=v_2(n_2)$$
as well as the one-parametric solution  of the latter set,
$$m_4=-13\,n_4/20=m_4(n_4)\,,\ \ \ \ \
v_4= -7\,n_4/20=v_4(n_4)\,.$$
With these results
it is now time to return to secular Eq.~(\ref{manine}) which is
exact but
manifestly nonlinear in the shifts.
The smallness of
the scaling parameter $\lambda$
enables us to
linearize such an equation simply by
omitting
all of its
${\cal O}(\lambda^4)$
components.
After the insertion of the
one-parametric solutions this yields the needed
dominant-order linear equation
 \be
  - 41.0904\,m_2(n_2)+54\,n_2- 33.4536\,v_2(n_2)=0\,.
 \label{kleor}
 \ee
This equation defines
the unique value of the first missing quantity
$n_2$ as a parameter
which varies with $\mu$ and $\nu$, i.e., with $\alpha$ and $\beta$,
as follows,
 \be
 n_2=
 n_2(\alpha,\beta) \approx 0.4033\,{\beta}^{2}+ 1.2099\,\alpha\,\beta+
 1.2099\,{\alpha}^{2}\,.
 \label{leor}
 \ee
The value of the second missing parameter $n_4$ remains unspecified.
In the light of  Lemma \ref{lemma2} we have a free choice
of this parameter,
with the admissible range of its variability
given by Eq.~(\ref{necesa}), i.e., by the two bounds
 \be
 {(\alpha+\beta)}^{4}+  {(\alpha+\beta)}^{2}\,\mu+
 (\alpha+\beta)\,\nu\,  >f(n_4) > \,
 {\alpha}^{4}+  {\alpha}^{2}\,\mu+
 \alpha\,\nu\,
 \label{lnecesa}
 \ee
where,
for the sake of brevity, we
used abbreviations
$$\mu=\mu(\alpha,\beta)
=2\,{\beta}^{2}+6\,{\alpha}^{2}+6\,\alpha\,\beta\,,$$
$$\nu=\nu(\alpha,\beta)=
8\,{\alpha}^{3}+12\,{\alpha}^{2}\,\beta+4\,\alpha\,{\beta}^{2}\,$$
and
$$
f(n_4)=f(n_4,\Omega)=
- 41.0904\,m_4(n_4)+54\,n_4- 33.4536\,v_4(n_4)+\Omega=
$$
$$
=3\,n_4\sqrt {949}+\Omega\approx 92.4175\,n_4+\Omega
$$
where $\Omega={\cal O}(\lambda^0)$ is a new auxiliary parameter.
Without
the use of the smallness of $\lambda$,
the exact value of this parameter could be written
in the form of
the following quadratic polynomial in the variable $n_4$,
 \be
\Omega=\Omega(n_4)=[m_2(n_2)+\lambda^2\,m_4(n_4)]\,
[v_2(n_2)+\lambda^2\,v_4(n_4)]\,.
 \label{qje4}
 \ee
Naturally, the use of this formula would return us back to the
exact but intrinsically nonlinear
recipe of Lemma \ref{lemma2}.
Now, after the restriction of attention to the vicinity of EP4,
we may use just the
leading-order approximation
$
\Omega = \Omega_0+{\cal O}(\lambda^2)$ with
the known and $n_4-$independent constant
$\Omega_0=m_2(n_2)\,v_2(n_2)$.
On this level of accuracy the
exact  but non-linear ${\cal D}^{[4]}-$specifying inequality
(\ref{lnecesa}) becomes approximate
but linear in $n_4$, i.e., explicit.

We may conclude that
the width  ${\cal O}(\lambda^4)$ of the
interval of variability of the shift $n$ is
non-vanishing. In fact, this width is
not too large, being
smaller than
the leading-order part of $n={\cal O}(\lambda^2)$
by the two orders of magnitude.
In the language of geometry this means that
the domain ${\cal D}$
is, near its fine-tuned EP4 extreme, sharply spiked.
Such a feature is
shared with  ${\cal D}$
of the preceding section and
seems generic
(cf. also,
in the simpler harmonic-oscillator-like
setting, the analogous result in paper \cite{maximal}).

\section{Discussion and summary}

Quasi-Hermitian formulation of quantum mechanics
was given its name in review \cite{Geyer}.
The origin of the approach has been attributed there to
the Dyson's paper \cite{Dyson}
which was devoted to the quantum many-body problem
and which,
later,
inspired a successful implementation of the formalism
in nuclear physics \cite{Jensen}.
From the point of view of mathematics
the authors of review \cite{Geyer}
succeeded in
circumventing the
Dieudonn\'{e}s' \cite{Dieudonne}
early criticism of the concept of quasi-Hermiticity
by
insisting on
the boundedness of the eligible Hamiltonians.
In this context the present
paper can be read as a well-motivated
proof of feasibility
of the study of
quantum systems described by
some less elementary
finite-dimensional quasi-Hermitian Hamiltonians.

From the historical point of view
the bounded-operator
constraint was discouraging. Not too surprisingly,
it seriously
diminishes the phenomenological
appeal of the theory \cite{Viola}.
Moreover,
the recommended transition
from the ``unacceptable''
differential operators
[sampled here by Eq.~(\ref{imcu})]
to non-Hermitian matrices
encountered a number of technical
obstacles \cite{Wilkinson}
and revealed multiple manifestations of
a user-unfriendlines of the models \cite{4x4,bh6x6}.

The current rebirth of enthusiasm over matrices
grew, step by step, out of two sources.
The first one originated from the symmetry-based discoveries
of
the exactly solvable matrix models
of an undeniable phenomenological relevance
({\it pars pro toto\,} let us mention the complexified
Bose-Hubbard
Hamiltonians of Refs.~\cite{Uwe,zaUweho,[37]}).
The second, complementary, more technical
source
of optimism
may be seen in a not quite expected
user-friendliness of the
matrices which are tridiagonal \cite{metricsaho}
or which are equal to the
direct sums of tridiagonal matrices \cite{preprint}.

Several encouraging results of the latter type
inspired also our present study.
In it, we managed to show that
the user-friendliness of the
tridiagonality assumption
seems decisive. We felt surprised by
the feasibility of the analysis of
our sufficiently realistic (i.e.,
not requiring the
equidistance of
the spectra)
and sufficiently flexible
(i.e., multiparametric)  real-matrix models
$H^{(N)}(\lambda)$.
We found that the computer-assisted
tractability of these models
need not necessarily be
enhanced,
at the cost of the loss of
the universality,
by any {\it ad hoc\,} simplifications.
In this sense our present
constructions
offered a natural continuation
of their equidistant-spectra-based
predecessors \cite{maximal}.
We managed to show that
the methods developed during the latter, much simpler constructions
remained equally well applicable also
after the
diagonal matrix elements of $H^{(N)}_0$
ceased to be kept equidistant.

Even less expectedly
we also demonstrated that the
tractability of the models
(i.e., e.g.,
the proofs and
localizations of their exceptional points, or
the specification of the boundaries of the
physical parametric domains ${\cal D}^{[N]}$)
also survived the purposeful removal of the
phenomenologically rather artificial
and purely technically motivated
assumption of ${\cal PT}-$symmetry
of the off-diagonal
part
of the Hamiltonians.
Such a part (called ``perturbation'')
was in fact considered, mostly, in the
strong-coupling dynamical regime.
For good reasons:
in such a regime a small change of the parameters
could cause the loss of the stability. Indeed,
close to the EPN extremes,
the corridors of stability appeared as narrow
as in the older, spectral-equidistance-based
arrangements. In both of these cases, the stability
near EPNs remains
equally fragile -- cf. also the occurrence of the same
feature
in many other models~\cite{Trefethen,Viola,fragile}.

Some of our results were expected:
the omission of symmetries
enhanced,
at a given dimension $N$,
the number of the degrees of freedom
(roughly speaking, by the factor of two).
The structure of the spectra
became richer, at the expense of
a fairly quick decrease of the feasibility
of the constructions with the growth of
the dimension beyond $N=6$.

Some of the other results proved
unpredictable but important:
in the non-equidistant-spectrum arrangement
the very existence of the maximal EPN degeneracies
came as a surprise.
The computer-assisted
symbolic-manipulation technique of their
localization
proved, in comparison with equidistant-spectrum
implementation in \cite{maximal},
equally efficient.

In our present less symmetric and
would-be more realistic square-well-inspired
models the loss of
the symmetries led, expectedly, to the non-existence of
non-numerical, closed-form
formulae and solutions at
the larger matrix dimensions $N \geq 6$.
At the same time,  the computer-assisted
constructions still remained
applicable, with the
feasibility
restricted only by the
capacity of the computer.
To an extent which seems,
in both of the equidistant- and non-equidistant-diagonal
scenarios, comparable.

A new symbolic-manipulation challenge
appeared when we decided to describe the
shape of the unitarity-compatible domains  ${\cal D}^{[N]}$
of the parameters
at which the spectrum remains real and non-degenerate.
Several satisfactory answers were given:
the determination of ${\cal D}^{[N]}$
proved exhaustive and non-numerical
not only at $N=3$
(see Lemma \ref{lemma1} --
this could really be good news for experimentalists)
but also at $N=4$
(see Lemma \ref{lemma2} --
this type of closed result,
albeit complicated in its form, could again be called unexpected).

Another, last but related
answer comprises the
explicit (at $N=3$ and $N=4$) or iterative (at $N=5$, cf. Appendix A.1)
or just conceptual and computer-mediated
(at $N=6$ and beyond, cf. Appendix A.2)
construction of the quantum-phase-transition boundaries
$\partial {\cal D}^{[N]}$, with emphasis upon their
leading-order narrow-corridor geometry
in the vicinity of the EPN dynamical extremes.

\section*{Appendix A. Localization of the
corridors at the larger $N$}


Far from EPN the
nonlinearities start playing a decisive technical role
at $N$ as small as $N=4$.
At $N \geq 5$ the consequent non-approximate algebraic specification
of the unitarity-compatible parametric
domain ${\cal D}^{[N]}$
becomes complicated.
Anyhow, whenever needed,
a more systematic approach to
approximations is still feasible and can serve the purpose,
in spite of a perceivably more lengthy form of its
presentation.

\subsection*{A.1.
Iterative construction at $N=5$}

In a way paralleling our preceding considerations
the quadruplet $m,n,v,w$
of the admissible shifts
is to be deduced
from the available secular equation
$\phi^{(5)}(E)=0$ using the exact evaluation
 $$
 \phi^{(5)}(E)=E^{\,5}+P\,E^{\,3}+Q\,E^{\,2}+R\,E+S
 $$
of the
secular polynomial, with the two non-numerical equations
   $$
   P= - w- v- n- m= \lambda^2\,\mu\,,
   \ \ \ Q= - 2\,v+ 11\,n- 21\,w+ 18\,m= \lambda^3\,\nu \,,
 $$
and with the numerical rest,
 $$
 R= nw+mw+ 189.33\,v+
 52.49\,n-
 %
 174.62\,m+mv-
 203.38\,w = \lambda^4\,\rho\,,$$
 $$
 S=  809.37\,m+ 3\,mw
+ 10\,nw- 16\,mv- 803.06\,w-
1575.05
\,n+ 1578.75\,v= \lambda^5\,\sigma\,.
 $$
Although we worked with the latter two
expressions in a high-precision arithmetics,
their display is shortened.
Under this convention the requirement $\phi^{(5)}(E)=0$
degenerates to the four
sets of equations, viz.,
$$
m+n+v+w=\mu\,\lambda^2
\,,\ \ \ \ \ \ \
18\,m+11\,n-2\,v-21\,w=\nu\,\lambda^3\,,
$$
$$
-174.62\,m+ 52.49\, n
+ 189.33\, v- 203.38\,w
=(\rho-\Omega)\,\lambda^4\,,
$$
$$
809.37\,m-1575.05\,n
+ 1578.75\, v- 803.06\,w
=-\Sigma\,\lambda^4+\sigma\,\lambda^5\,.
$$
Here, $\sigma$ is
to be treated as a temporarily undetermined, variable quantity
which will only enter the game,
at the very end of the construction, when
sampling the full-secular-polynomial  ${\cal O}(\lambda^{5})$
contribution evaluated
at its trial-and-error energy minima and maxima.
Also the other two auxiliary new functions $\Omega=\Omega_0+
\lambda\,\Omega_1+{\cal O}(\lambda^2)$ and $\Sigma=\Sigma_0+
\lambda\,\Sigma_1+{\cal O}(\lambda^2)$ (i.e.,
on this level of approximation, the four new real parameters)
should temporarily be kept
indeterminate.

Once we relocated the nonlinearities in parameters,
the dominant-order system becomes linear.
Its solution becomes routine and, as before, its
order-by-order form can be
sought via the
set of four ansatzs,
 $$
 m=\lambda^2\,m_2+\lambda^3\,m_3+\lambda^4\,m_4+\lambda^5\,m_5
 \,,\ \ \
 n=\lambda^2\,n_2+\lambda^3\,n_3+\lambda^4\,n_4+\lambda^5\,n_5
 \,,
 $$
 $$
 v=\lambda^2\,v_2+\lambda^3\,v_3+\lambda^4\,v_4+\lambda^5\,v_5
 \,,\ \ \
 w=\lambda^2\,w_2+\lambda^3\,w_3+\lambda^4\,w_4+\lambda^5\,w_5\,.
 $$
Due to the partially order-separated structure,
our secular equation may be re-arranged as a set of 16
linear equations for the 16 relevant coefficients
$m_2, n_2, \ldots, w_5$.
This set of equations
starts from an elementary four-equation subset
$$
m_2+n_2+v_2+w_2=\mu
\,,\ \ \ \ \ \ \
m_3+n_3+v_3+w_3=0
\,,\ \ \ \ \ \ \
m_4+n_4+v_4+w_4=0
\,,\ \ \ \ \ \ \
m_5+n_5+v_5+w_5=0
$$
and from another quadruplet of linear equations
with integer coefficients,
$$
18\,m_2+11\,n_2-2\,v_2-21\,w_2=0\,,
\ \ \ \ \
18\,m_3+11\,n_3-2\,v_3-21\,w_3=\nu\,,
 $$
 $$
18\,m_4+11\,n_4-2\,v_4-21\,w_4=0\,,
\ \ \ \ \
18\,m_5+11\,n_5-2\,v_5-21\,w_5=0\,.
$$
The next item is the quadruplet of the
numerically represented equations
$$
-174.62\,m_j+ 52.49\, n_j
+ 189.33\, v_j- 203.38\,w_j
=(\rho-\Omega_0)\,\delta_{j,4}-\Omega_1\,\delta_{j,5}
$$
with the usual Kronecker $\delta_{j,4}$ and  $j=2,3,4,5$.
The last set, with    $j=2,3,4,5$, has the similar structure,
$$
809.37\,m_j-1575.05\,n_j
+ 1578.75\, v_j- 803.06\,w_j
= -\Sigma_0\,\delta_{j,4}+(\sigma-\Sigma_1)\,\delta_{j,5}\,.
$$
Only after we make use of the linearity and
construct all of the $\sigma-$, $\Omega-$ and $\Sigma-$dependent
but $\lambda-$independent
closed-form-coefficient solutions
$m_2(\mu,\nu,\rho,\sigma,\Omega,\Sigma), \ldots, w_5(\mu,\nu,\rho,\sigma,\Omega,\Sigma)$
of all of the 16 equations, we become prepared to
recall
the correct explicit definition of
$$
\Omega=\lambda^{-4}\,(m\,v+m\,w+n\,w)=\Omega_0+
\lambda\,\Omega_1+{\cal O}(\lambda^2)\,
$$
leading to the necessity of a de-linearization insertion of
$$ \Omega_0=
m_2\,v_2+m_2\,w_2+n_2\,w_2 \,,
\ \ \ \ \Omega_1=
m_2\,v_3+m_3\,v_2+m_2\,w_3+m_3\,w_2+n_2\,w_3+n_3\,w_2 \,.
$$
Similarly, we have to deal with the de-linearizing reinsertions of
$$
\Sigma=\lambda^{-4}\,(-16\,m\,v+3\,m\,w+10\,n\,w)=\Sigma_0+
\lambda\,\Sigma_1+{\cal O}(\lambda^2)\,,
 $$
with
  $$
  \Sigma_0=-16\,m_2\,v_2+3\,m_2\,w_2+10\,n_2\,w_2\,
$$
and
 $$
  \Sigma_1=-16\,m_2\,v_3+3\,m_2\,w_3+10\,n_2\,w_3
-16\,m_3\,v_2+3\,m_3\,w_2+10\,n_3\,w_2\,.
$$
Fortunately, all of these four reinsertions
introduce the corrections of subdominant orders
so that their inclusion can be performed, whenever
asked for,
iteratively.

\subsection*{A.2.
Iterative construction at $N=6$}

By a return to the
last, $N=6$ sample of our Hamilonians
we wish to indicate that
using the computer-assisted manipulations
the approximate, linearization-based
determination of
the domains of unitarity
${\cal D}^{[N]}$
remains also feasible
at a few further, not too small matrix dimensions.
We just recall that at $N=6$ we
determined,
in section \ref{sepeti}, the coordinates
of the EP6 extreme.
They were denoted as
$A^{(EP6)},B^{(EP6)},C^{(EP6)},D^{(EP6)}$ and $F^{(EP6)}$
and interpreted as starting points
of a move inside
the
unitarity-compatible
vicinity ${\cal D}^{[6]}$
of this extreme,
with
$A=A^{(EP6)}-m,B=B^{(EP6)} -n,C
=C^{(EP6)}-v,D=D^{(EP6)}-w$ and $F=F^{(EP6)}-y$
in Hamiltonian (\ref{mo6}).

The real shifts $m, n, \ldots, y$ will be now assumed
small, ${\cal O}(\lambda^2)$.
The smallness assumption
restricts our attention to the
spiked part of ${\cal D}^{(6)}$
forming
a not too broad parametric corridor in which
the bound state energies would still
stay real and non-degenerate,
i.e., in which the quantum system in question would
remain observable and unitary.
In a preparatory step
of such a climax of the project
the reparametrization of couplings $A, B, \ldots,F$
in terms of
the real ${\cal O}(\lambda^2)$ shifts $m, n, v, w$ and $y$
yields again the reparametrized secular polynomial
 $$
  \phi^{(6)}(s)=
{s}^{6}+ \left( -w-y-v-n-m \right) {s}^{4}+ \left( 28\,m+21\,n-11\,w+8
\,v-36\,y \right) {s}^{3}+
 $$
 $$
 + \left( - 469.95\,m+my-
 684.78\,y- 71.22\,n
  +vy+nw+ny+mv+
 352.55\,v+ 361.95\,w+mw \right) {s}^{2}
+
 $$
 $$
 +\left( -36\,mv- 4589.07\,n- 7227.48,y+8
\,my
+ 10332.26\,w+28\,vy+15\,ny
+ 4971.82\,m
-10\,nw - \right .
$$
$$
\left . -17\,mw-
 3799.19\,v \right) s-
$$
$$
 -
 33366.21\,y
+
 68688.06\,n-
 97336.49\,v-26175.19\,m
+
 77974.33\,w-
$$
$$
-375\,nw
+146.73\,my+
 428.05\,mv+ 227.40\,vy-200\,mw-mvy=
 $$
 $$
 =s^6+Ps^4+Qs^3+Rs^2+Ss+T\,.
 $$
We will
rescale
 $$
 P=\mu\,\lambda^2\,,
 \ \ \
 Q=\nu\,\lambda^3\,,
 \ \ \
 R=\rho\,\lambda^4\,,
 \ \ \
 S=\sigma\,\lambda^5\,,
 \ \ \
 T=\tau\,\lambda^6\,,
 $$
and expand
 $$
 m
 =\lambda^2\,m_2+\lambda^3\,m_3+\lambda^4\,m_4+\lambda^5\,m_5
 +\lambda^6\,m_6 +\ldots
 $$
 $$
 n
 =\lambda^2\,n_2+\lambda^3\,n_3+\lambda^4\,n_4+\lambda^5\,n_5
 +\lambda^6\,n_6 +\ldots
 $$
etc.
This enables us to formulate our problem
as a search for a criterion
imposed upon
the correct shifts $m,n,\ldots,y$
which would guarantee
(i.e., which would provide a sufficient condition of)
the reality of the spectrum.

In the light of our preceding experience
we will perform again two
{\it ad hoc\,} reparametrizations
expressing, firstly,
the shifts in terms of the
secular-polynomial coefficients
(which would yield the functions
$m= m(\lambda,\mu,\ldots,\tau)$, etc)
and, secondly, these
secular-polynomial coefficients themselves
in terms of the coordinates of
the extremes of the secular polynomial.
Thus, another set of new functions
(viz., $\mu= \mu(\alpha,\beta,\ldots)$, etc)
will enter the mathematical scene.

Our key idea remains the same as above, treating
the
reality of the whole $N-$plet of energies as
requiring the
existence
(plus some further properties)
of  as many as $N-1$
distinct real extremes
of $\phi^{(6)}(x)$.
This opens the possibility
of treating the {\em variable\,} positions of
these extremes (at $x_1=a$,  $x_2=a+b$, etc)
as an optimal initial dynamical-information input.

Our task again lies in
the reconstruction of the
related non-Hermitian matrix Hamiltonian
(with elements defined in terms of the shifts $m,n,\ldots$)
which would generate unitary evolution.
For the sake of brevity we will only sketchily
mention the separate steps
of the $N=6$ construction
in full detail, skipping those
which remain fully analogous to their $N<6$ predecessors.
Thus, we will leave the detailed
geometric interpretation of
the initial,
level-spacing parameters $a,b, \ldots$
and of their respective $\lambda^2-$rescaled representations
$\alpha,\beta, \ldots$ to the readers. We will also omit,
in a self-explanatory notation, the
details of transition from  $a,b, \ldots$ or
$\alpha,\beta, \ldots$ to $P,Q, \ldots$
or $\mu,\nu, \ldots$. Indeed,
such a step
(based on the factorization of
the derivative of the secular polynomial) was
also sufficiently thoroughly explained above.

\begin{table}[h]
\caption{Coefficients ${C^{(N)}_{2,n}}$
in Eq.~(A1)}\label{zp4}

\vspace{2mm}

\centering
\begin{tabular}{||cc||cccccc||}
\hline \hline
\multicolumn{2}{||r||}{$\ \ \ \ n$}&1&2&3&4&5&6\\
\multicolumn{2}{||l||}{$N^{}$}&\multicolumn{6}{c||}{{\mbox{}}}\\
\hline
\hline
2&&0&&&&&\\
3&&4&-3&&&&\\
4&&10&3&-10&&&\\
5&&18&11&-2&-21&&\\
6&&28&21&8&-11&-36&\\
7&&40&33&20&1&-24&-55\\
\hline
 \hline
\end{tabular}
\end{table}

What only remains less clear
is the last-step transition from  $\mu,\nu, \ldots$
to  $m,n, \ldots$ where
the nonlinearity
of the underlying
system of
coupled algebraic equations
is a decisive
technical
obstacle.
In the language of mathematics the task is to
reconstruct the
Hamiltonian, i.e., the $(N-1)-$plet of the small and real shifts,
i.e.,  at $N=6$, of $m, n, v, w$ and $y$.
The tools are the inspection of polynomial $\phi^{(6)}(s)$
and
the use of the smallness of $\lambda$. In
a way paralleling the construction at $N=5$
the vanishing of the secular polynomial
$\phi^{(N)}(E)=0$
(in some of the present comments,
the value of $N$ need not be just five
or six)
may be treated as an $(N-1)$ by $(N-1)$
matrix inversion problem
 $$\ \ \ \ \ \ \ \ \ \ \ \ \ \ \ \ \ \ \ \ \ \ \ \
 C^{(N)} \,\vec{z}=\vec{\omega}\,\ \ \ \ \ \ \ \ \ \ \ \
 \ \ \ \ \ \ \ \ \ \ \ \ \ \ \ \ \ \ \ \ \ \ \ \
 \ \ \ \ \ \ \ \ \ \ \ \ (A1)
 $$
with a few redundant, auxiliary
parameters in
the right-hand-side column vector $\vec{\omega}$.
The matrix $C^{(N)} $ of the system
(with coefficients sampled in our preceding considerations)
is $\lambda-$independent.
Its first row
is,
and any $N$,
trivial,
$C^{(N)}_{1,j}=1$ at $j=1,2,\ldots,N$.
Similar observation also applies to the second row.
Table \ref{zp4}
samples this row
up to $N =7$ but
the extrapolation
of these $(N-1)-$plets of integers
to any dimension $N$ is obvious (the elementary exercise
of derivation of
explicit formula is left to the readers).
The rest of the matrix  $C^{(N)} $
remains numerical and strongly $N-$dependent
(its $N \leq 6$ samples
could be extracted from preceding sections).

Even at $N=6$,
equation (A1)
already becomes sufficiently
rich to
illustrate the general pattern.
In its respective unknown- and known-vector parts
 $$
 \vec{z}=
 \left (\ba
 m\\
 n\\
 v\\
 w\\
 y
 \ea \right )\,,
 \ \ \ \ \ \
 \vec{\omega}=
 \left (\ba
 P\\
 Q\\
 R-\Omega\\
 S-\Sigma\\
 T-\Pi
 \ea \right )\,
 $$
the latter one contains
the five dynamical-input
parameters $P, Q, \ldots, T$ and
the
three auxiliary functions
$\Omega={\cal O}(\lambda^4)$, $\Sigma={\cal O}(\lambda^4)$
and $\Pi={\cal O}(\lambda^4)$. All of them are
polynomials with coefficients kept
indeterminate. This renders the
technically straightforward
matrix-inversion solution feasible,
 $$
 \vec{z}=\left [
 C^{(N)} \right ]^{-1}\,
 \vec{\omega}\,.
 $$
Along the lines known from the $N=5$ predecessor
of the model
this leads to the order-by-order
decomposition of
solution
into its separate $\lambda-$independent
components.
The process
defines the set of 25
coefficients $m_2, n_2, \ldots, y_6$
which remain dependent on
the multiplet of auxiliary parameters.

As long as the nonlinearity
of the full-fledged equations
has
the form of
corrections,
they may be incorporated
in an iterative manner.
The nonlinearity
survives, in disguise, via
the reinstalled constraints
$$
\Omega=m\,v+m\,w+m\,y+
 n\,w+n\,y+v\,y\,,
$$
and
$$
\Sigma=-36\,m\,v-17\,m\,w+8\,m\,y-10\,
 n\,w+15\,n\,y+28\,v\,y\,.
 $$
It is only necessary to add that
the role and status of
$\Pi$ is different.
Firstly, purely formally, the formula for $\Pi$ is
already too long
to be displayed in print:
One of the reasons is
that is already
contains the triple-product contribution
equal to $-m\,w\,y$, and another one is that
many of its
coefficients cease to be integers.
After all, the
fully explicit power-series forms of the
other two
expressions
(and, hence,
the whole $N=6$ construction)
would be also too complicated.
It is only necessary to add that although
their display in print would be entirely formal,
all of these expressions
may still be  easily stored in the computer.
This means that
even at $N=6$,
the dominant-order specification of the
boundaries of the
parametric corridor admitting
unitary evolution remains
feasible.

\subsection*{Acknowledgments}

The author acknowledges the financial support from the
Excellence project P\v{r}F UHK 2020.

\subsection*{Data Availability}

Data sharing not applicable – no new data generated.


\end{document}